\pdfoutput=1
\documentclass[letterpaper, 10 pt, conference]{ieeeconf}  % Comment this line out if you need a4paper

\IEEEoverridecommandlockouts                              % This command is only needed if 
\usepackage{graphicx,algorithm,algorithmic,amsthm,amsmath,amsfonts,verbatim,mathtools,enumerate,bm,hyperref,dsfont,color,tikz, pgfplots,datatool,filecontents, siunitx,subcaption}
\usepgfplotslibrary{fillbetween}
\pgfplotsset{compat=1.18}
\usepackage{subcaption}
\hypersetup{
    urlcolor=cyan
    }
\allowdisplaybreaks

\newcommand{\tr}{\mathop{\rm tr}}

\newcommand{\diag}{\mathop{\rm diag}}

\newcommand{\lea}{\texttt{L}}
\newcommand{\fow}{\texttt{F}}

\newcommand{\norm}[1]{\left\lVert#1\right\rVert}
\newcommand{\mnorm}[1]{{\left\vert\kern-0.25ex\left\vert\kern-0.25ex\left\vert #1 
    \right\vert\kern-0.25ex\right\vert\kern-0.25ex\right\vert}}

\newtheorem{proposition}{Proposition}

\newcommand{\eg}{{\it e.g.}}
\newcommand{\ie}{{\it i.e.}}

\definecolor{yygreen}{rgb}{0 1 0}
\definecolor{yypurple}{rgb}{0.4940 0.1840 0.556}
\definecolor{yyred}{rgb}{1 0 0}
\definecolor{yyblue}{rgb}{0 0 1}

\DTLloaddb[noheader=false]{LF_type1}{LF_t1.dat}
\DTLloaddb[noheader=false]{LF_type2}{LF_t2.dat}
\DTLloaddb[noheader=false]{LF_type3}{LF_t3.dat}
\DTLloaddb[noheader=false]{LRoad_type1}{LRoad_t1.dat}
\DTLloaddb[noheader=false]{LRoad_type2}{LRoad_t2.dat}
\DTLloaddb[noheader=false]{LRoad_type3}{LRoad_t3.dat}

\title{\LARGE \bf Active Inverse Learning in Stackelberg Trajectory Games
}

\author{William Ward, Yue Yu, Jacob Levy, Negar Mehr, David Fridovich-Keil, and Ufuk Topcu% <-this % stops a space
\thanks{W. Ward, J. Levy, D. Fridovich-Keil, and U. Topcu are with the Oden Institute for Computational Engineering and Sciences at the University of Texas at Austin, Austin, TX, 78712, USA (emails: wward@utexas.edu,\, jake.levy@utexas.edu,\, dfk@utexas.edu,\, utopcu@utexas.edu). Y. Yu is with the Department of Aerospace Engineering and Mechanics at the University of Minnesota Twin Cities, Minneapolis, MN 55455 (email: yuey@umn.edu). N. Mehr is with the Department of Mechanical Engineering at the University of California Berkeley, Berkeley, CA, 94720, USA (email: negar@berkeley.edu).}%
}

\begin{document}

\maketitle
\thispagestyle{empty}
\pagestyle{empty}

%%%%%%%%%%%%%%%%%%%%%%%%%%%%%%%%%%%%%%%%%%%%%%%%%%%%%%%%%%%%%%%%%%%%%%%%%%%%%%%%
\begin{abstract}
Game-theoretic inverse learning is the problem of inferring a player's objectives from their actions. We formulate an inverse learning problem in a Stackelberg game between a leader and a follower, where each player's action is the trajectory of a dynamical system. We propose an active inverse learning method for the leader to infer which hypothesis among a finite set of candidates best describes the follower's objective function. Instead of using passively observed trajectories like existing methods, we actively maximize the differences in the follower's trajectories under different hypotheses by optimizing the leader's control inputs. Compared with uniformly random inputs, the optimized inputs accelerate the convergence of the estimated probability of different hypotheses conditioned on the follower's trajectory. We demonstrate the proposed method in a receding-horizon repeated trajectory game and simulate the results using virtual TurtleBots in Gazebo.
\end{abstract}

%%%%%%%%%%%%%%%%%%%%%%%%%%%%%%%%%%%%%%%%%%%%%%%%%%%%%%%%%%%%%%%%%%%%%%%%%%%%%%%%
\section{Introduction}
Learning to predict human behavior is a critical challenge in human-robot interaction. It enables robots to customize their strategies in various applications, including assisted driving \cite{shia2014semiautonomous,mehr2016inferring}, traffic management \cite{mehr2023maximum,sadigh2016information}, and, in general, mitigating conflicts in human-in-the-loop robotic systems.    

Game-theoretic inverse learning helps robots explain and predict human behavior in noncooperative interactions where humans actively optimize only their own objectives \cite{waugh2010inverse,kuleshov2015inverse,bertsimas2015data,peters2021inferring,molloy2022inverse,yu2022inverse,li2023cost,mehr2023maximum}. The idea is to first model humans' objectives as parameterized functions, then infer the parameter value such that the corresponding game-theoretic strategies---such as Nash or Stackelberg equilibrium strategies---match the humans' actions in a dataset. Game-theoretic inverse learning is a necessary step in understanding human-robot interactions \cite{ling2018game,chen2023soft,mehr2023maximum} and designing incentives for multiagent systems \cite{bacsar1984affine,nisan2015algorithmic}.  
 
The existing game-theoretic inverse learning methods are \emph{passive}. In particular, these methods record the dataset of human actions before and independently of the inference process. Hence, some actions in the recorded dataset are uninformative for inference purposes, or simply redundant. As a result, in some settings, passive inverse learning is insufficient for rapid inferencing and online real-time decision-making \cite{sadigh2016information}. %\david{this claim is dubious to me, without specific experimental evidence. Perhaps it is better to say that there are examples where passive inference is just not sufficient in other settings, e.g. some of Dorsa Sadigh or Mykel Kochenderfer's work}. 
  
In contrast to passive inverse learning, active inverse learning helps robots infer human intentions in cooperative interactions by actively provoking informative human responses. For example, when learning objectives that explain a human's ranking or rating of presented options, active inverse learning methods first provoke informative human responses and record them in the dataset, then infer the human's objective function, and repeat this process if necessary \cite{lopes2009active,akrour2012april,sadigh2017active,christiano2017deep,daniel2014active}. These methods ensure that the human's actions are informative by maximizing the volume removed from the hypothesis space \cite{sadigh2017active,biyik2018batch,biyik2019green,palan2019learning,katz2019learning,basu2019active} or by maximizing the information gain \cite{cohn2011comparing,daniel2015active,sadigh2016information,cui2018active,erdem2020asking,myers2022learning}. By integrating dataset updates with inference, active inverse learning provides practical solutions for inferring human intentions from limited interactions. 

Despite its successes, active inverse learning still requires investigation in noncooperative interactions. The existing active inverse learning methods rely on querying humans who volunteer informative responses. In contrast, humans in noncooperative interactions only take actions that optimize their own objectives, regardless of whether or not the actions are informative. Therefore, how to provoke informative actions from noncooperative humans that reveal their objectives is, to our best knowledge, still an open question. %\david{I think it would be nice if this paragraph could give a brief example of how the noncooperative aspect really necessitates a new approach} 

We study Stackelberg games, a specific class of noncooperative interactions between two players: a leader and a follower. In a Stackelberg game, the leader chooses an action, and then the follower responds by observing the leader's behavior and optimizing its objective function. 

We formulate an inverse learning problem in a Stackelberg game where a rational leader, such as a robot, infers which hypothesis among finitely many candidates best explains the behavior of a boundedly rational follower, such as a human. This problem is particularly relevant in shared autonomy, \eg, when an autopilot must infer the type of a newly encountered human driver. 

We model each player's action as the trajectory of a linear time-invariant system. The follower tracks a linear function of the leader's trajectory using a maximum-entropy linear quadratic regulator. The regulator contains a parameterized objective function and models bounded rationality in human decision-making \cite{mehr2023maximum}. The leader determines which hypothesis is most likely using the probability of each hypothesis conditioned on the follower's state trajectory.

% The follower tracks a linear function of the leader's trajectory---similar to how a human driver tracks the trajectory recommended by an autopilot---using a maximum-entropy linear quadratic regulator---which contains a parameterized objective function---that models bounded rationality in human decision-making \cite{mehr2023maximum}. 

We propose an active inverse learning method to provoke informative trajectories from the follower by optimizing the leader's inputs. In this optimization, we maximize the differences in the follower's trajectory distributions under different hypotheses. We show that this optimization is a difference-of-convex program \cite{horst1999dc}, which can be solved efficiently via the convex-concave procedure \cite{lipp2016variations}. We evaluate the performance of the proposed method in a receding-horizon repeated trajectory game. Compared with random inputs, the leader inputs provided by our method accelerate the convergence of the probability of different hypotheses conditioned on the follower's trajectory.

\noindent\textit{Notation:} We let \(\mathbb{R}\), \(\mathbb{R}_{\geq 0}\), and \(\mathbb{N}\) denote the set of real numbers, nonnegative real numbers, and nonnegative integers, respectively. We let \(\mathbb{S}_{\succeq 0}^n\) and \(\mathbb{S}_{\succ 0}^n\) denote the set of \(n\) by \(n\) symmetric positive semidefinite and positive definite matrices, respectively. For any \(x\in\mathbb{R}^n\), we let \(\norm{x}\coloneqq \sqrt{x^\top x}\), \(\norm{x}_1\coloneqq \sum_{i=1}^n |x_i|\), \(\norm{x}_\infty\coloneqq \max_{1\leq i\leq n} |x_i|\), and  \(\norm{x}_A^2\coloneqq x^\top A x\) for all \(A\in\mathbb{S}_{\succeq 0}^n\). We let \(0_n\) denote the \(n\)-dimensional zero vector; \(I_n\) and \(0_{n\times n}\) denote the \(n\) by \(n\) identity and zero matrix, respectively. We let \(\mathcal{N}(\mu, \Sigma)\) denote the Gaussian distribution with mean \(\mu\in\mathbb{R}^n\) and variance \(\Sigma\in\mathbb{S}_{\succ 0}^n\). Given \(n_1, n_2\in\mathbb{N}\), we let \([n_1, n_2]\) denote the set of integers between \(n_1\) and \(n_2\). Given \(a_i\in\mathbb{R}^n\) for all \(i\in\mathbb{N}\), we let \(a_{i:j}\coloneqq \begin{bmatrix} a_i^\top & a_{i+1}^\top & \cdots & a_j^\top\end{bmatrix}^\top\) for all \(i< j\), \(i, j\in\mathbb{N}\).

\section{Linear quadratic Stackelberg trajectory games}
\label{sec: stackelberg}
We introduce a Stackelberg game between a rational leader, such as a robot, and a boundedly rational follower, such as a human with noisy behavior. The players' actions are trajectories of stochastic linear time-invariant systems. 

\subsection{The dynamics of the players' systems}
We assume that the leader's state evolves according to the following discrete-time linear time-invariant dynamics:
\begin{equation}\label{sys: leader lti}
    x_{t+1}^\lea=A^\lea x_t^\lea+B^\lea u_t^\lea+w_t^\lea
\end{equation}
for all \(t\in\mathbb{N}\), where \(x_t^\lea\in\mathbb{R}^{n_\lea}\), \(u_t^\lea\in\mathbb{R}^{m_\lea}\), and \(w_t^\lea\in\mathbb{R}^{n_\lea}\) are the state, input, and disturbance of the system at time \(t\in\mathbb{N}\), respectively; \(A^\lea\in\mathbb{R}^{n_\lea\times n_\lea}\) and \(B^\lea \in\mathbb{R}^{n_\lea\times m_\lea}\) are the leader's system parameters.  In our experiments, Equation \eqref{sys: leader lti} to approximates the dynamics of a ground robot. %\david{well.. i mean that's not quite right. rovers and drones are most likely differentially flat, but most kinematic / dynamic models are not linear} 

Similarly, the follower's state evolves according to the following dynamics:
\begin{equation}\label{sys: follower lti}
    x_{t+1}^\fow=A^\fow x_t^\fow+B^\fow u_t^\fow+w_t^\fow
\end{equation}
for all \(t\in\mathbb{N}\), where \(x_t^\fow\in\mathbb{R}^{n_\fow}\), \(u_t^\fow\in\mathbb{R}^{m_\fow}\), and \(w_t^\fow\in\mathbb{R}^{n_\fow}\) denote the state, input, and disturbance of the system at time \(t\in\mathbb{N}\), respectively; \(A^\fow\in\mathbb{R}^{n_\fow\times n_\fow}\) and \(B^\fow \in\mathbb{R}^{n_\fow\times m_\fow}\) are the follower's system parameters. 

Throughout, we assume that the disturbances in the leader's system are independent, identically distributed Gaussian vectors. Similarly, the disturbances in the follower's system are independent, identically distributed Gaussian vectors. In other words, there exists \(\Omega^\lea\in\mathbb{S}^{n_\lea}_{\succ 0}\) and \(\Omega^\fow\in\mathbb{S}^{n_\fow}_{\succ 0}\) such that, for any \(t\in\mathbb{N}\), we have
\begin{equation}\label{eqn: gaussian disturbance}
    w_t^\lea\sim\mathcal{N}(0_{n_\lea}, \Omega^\lea), \, w_t^\fow\sim\mathcal{N}(0_{n_\fow}, \Omega^\fow).
\end{equation}

\subsection{The players' objectives}
We assume that the follower's objective is to track a linear function of the leader's trajectory. In particular, we let \(M^\fow\in\mathbb{R}^{n_\fow\times n_\lea}\) denote the parameters of a linear transformation that maps the leader's internal state to an output reference observable to the follower. Let \(x_{0:\tau}^\lea\) denote a leader trajectory of length \(\tau\in\mathbb{N}\). 
The follower's objective is to simultaneously track the corresponding output trajectory \(\{M^\fow x^\lea_0, M^\fow x^\lea_1, \ldots, M^\fow x^\lea_\tau\}\) and minimize its input efforts. 

We assume that the follower is boundedly rational and chooses its input according to the maximum entropy principle. Consequently, the distribution of \(u_t^\fow\) conditioned on \(x_t^\fow\) is Gaussian, \ie, \(\mu_t^\fow|x_t^\fow\sim \mathcal{N}(\mu_t, \Sigma_t)\) for some \(\mu_t\in\mathbb{R}^{m_\fow}\) and \(\Sigma_t\in\mathbb{R}^{m_\fow\times m_\fow}\) \cite{mehr2023maximum}. In particular, \((\mu_{0:\tau-1}, \Sigma_{0:\tau-1})\) is optimal for the following stochastic trajectory optimization problem:    
\begin{equation}\label{opt: f-lqr}
    \begin{array}{ll}
    \underset{\mu_{0:\tau-1}, \Sigma_{0:\tau-1}}{\mbox{minimize}} & \sum_{t=0}^{\tau} \mathds{E}\left[ \frac{1}{2}\norm{x_t^\fow-M^\fow x_t^\lea}_{Q^\fow}^2\right]\\
    &+\frac{1}{2}\sum_{t=0}^{\tau-1}\left(\mathds{E}\left[ \norm{u_t^\fow}_{R^\fow}^2 \right]-\log\det\Sigma_t\right) \\
    \mbox{subject to} & x_{t+1}^\fow=A^\fow x_t^\fow+B^\fow u_t^\fow+w_t^\fow, \, x_0^\fow=\hat{x}_0^\fow,\\
    & u_t^\fow|x_t^\fow\sim\mathcal{N}(\mu_t, \Sigma_t), \, w_t^\fow\sim\mathcal{N}(0_{n_\fow}, \Omega^\fow),\\
    &t\in[0, \tau-1],
    \end{array}
\end{equation}
where \(\hat{x}^{n_\fow}_0\in\mathbb{R}^{n_\fow}\) is the initial state of the follower's system, \(\mathds{E}[\cdot]\) denotes the expectation; \(Q^\fow\in\mathbb{S}^{n_\fow}_{\succeq 0}\), and \(R^\fow\in\mathbb{S}^{m_\fow}_{\succ 0}\) are the follower's cost parameters. The objective function in optimization~\eqref{opt: f-lqr} captures boundedly rational human decisions: it is noisy but centers around a cost-minimizing rational decision \cite{VonNeumann1944}, \cite{Duncan1959}, \cite{ziebart2008maximum}, \cite{wulfmeier2016maximum}, \cite{wu2020efficient}.     % \david{I think Negar will be able to point you toward good literature to cite in justifying this MaxEnt model of bounded rationality. If not I can also dig up some refs}  

The following proposition provides a closed-form expression for the solution of optimization~\eqref{opt: f-lqr}. 

\begin{proposition}\label{prop: follower riccati}
Let 
\begin{subequations}\label{eqn: follower DP}
    \begin{align}
    &P^\fow_{\tau}=Q^\fow,\, P^\fow_t = Q^\fow+(A^\fow)^\top P^\fow_{t+1} E^\fow_t,\\
    & F^\fow_t=B^\fow(R^\fow+(B^\fow)^\top P^\fow_{t+1}B^\fow)^{-1} (B^\fow)^\top,\\
    & E^\fow_t= A^\fow-F^\fow_t P^\fow_{t+1}A^\fow,\\
    &q_\tau^\fow=-Q^\fow M^\fow x_\tau^\lea, \\ 
    &q_t^\fow= (E^\fow_t)^\top q_{t+1}^\fow- Q^\fow M^\fow x_t^\lea,
    \end{align}
\end{subequations}
for all \(t\in[0,\tau-1]\). Given \(x^\lea_{0:\tau}\), \((\mu_{0:\tau-1}, \Sigma_{0:\tau-1})\) is optimal for optimization~\eqref{opt: f-lqr} if and only if
\begin{subequations}\label{eqn: mean & var for follower input}
    \begin{align}
        \Sigma_t & = (R^\fow+(B^\fow)^\top P^\fow_{t+1}B^\fow)^{-1},\\
        \mu_t & = -\Sigma_t(B^\fow)^\top (P^\fow_{t+1}A^\fow x_t^\fow+q_{t+1}^\fow).  
    \end{align}
\end{subequations}
for all \(t\in[0, \tau-1]\). Furthermore, if the constraints in \eqref{opt: f-lqr} hold, then \(x_t^\fow\sim\mathcal{N}(\xi_t, \Lambda_t)\) for all \(t\in[0, \tau]\), where\begin{subequations}\label{eqn: follower state traj}
    \begin{align}
        \xi_{t+1}&=E^\fow_t\xi_t-F^\fow_t q_{t+1}^\fow, \label{eqn: follower state mean}\\
        \Lambda_{t+1}&= E^\fow_t\Lambda_t(E^\fow_t)^\top +F^\fow_t+\Omega^\fow, \label{eqn: follower state var}
    \end{align}
\end{subequations}
for all \(t\in[0, \tau-1]\), with \(\xi_0=\hat{x}^\fow_0\) and \(\Lambda_0=0_{n_\fow\times n_\fow}\).
\end{proposition}
\begin{proof}
    See Appendix.
\end{proof}

The leader's objective is to minimize a cost function that jointly depends on the follower's expected trajectory and the leader's trajectory. To this end, we assume that the leader acts rationally and chooses its trajectory \(x^\lea_{0:\tau}\) as a solution to the following trajectory optimization problem: 
\begin{equation}\label{opt: l-lqr}
    \begin{array}{ll}
    \underset{u_{0:\tau-1}^\lea}{\mbox{minimize}} & \mathds{E}[f(x^\lea_{0:\tau}, x^\fow_{0:\tau})]+  g(u^\lea_{0:\tau})\\
    \mbox{subject to} & x_{t+1}^\lea=A^\lea x_t^\lea+B^\lea u_t^\lea+w_t^\lea,\, x_0^\lea = \hat{x}^\lea_0,\\
     & x_{t+1}^\fow=A^\fow x_t^\fow+B^\fow u_t^\fow+w_t^\fow, \, x_0^\fow=\hat{x}_0^\fow,\\
    &  w_t^\lea \sim\mathcal{N}(0_{n_\lea}, \Omega^\lea), \, w_t^\fow\sim\mathcal{N}(0_{n_\fow}, \Omega^\fow),\\
    & u_t\in\mathbb{U},\,u_t^\fow|x_t^\fow\sim\mathcal{N}(\mu_t, \Sigma_t),\, t\in[0, \tau-1], \\
    & (\mu_{0:\tau-1}, \Sigma_{0:\tau-1}) \text{ is optimal for \eqref{opt: f-lqr},}
    \end{array}
\end{equation}
where \(\hat{x}^\lea_0\in\mathbb{R}^{n_\lea}\) is the initial state of the leader's system and \(\mathbb{U}\subset\mathbb{R}^{m_\fow}\) is the set of feasible leader inputs at each time. Furthermore, \(f:\mathbb{R}^{(\tau+1)n_\lea}\times \mathbb{R}^{(\tau+1)n_\lea}\to\mathbb{R}\) is the leader's cost function that jointly depends on the leader's state trajectory \(x^\lea_{0:\tau}\) and the follower's state trajectory \(x^\fow_{0:\tau}\); 
\(g:\mathbb{R}^{\tau m_\lea}\to\mathbb{R}\) is a cost function that only depends on the leader's input trajectory. By choosing different functions for \(f\) and \(g\), optimization~\eqref{opt: l-lqr} achieves different trade-offs between optimizing the leader and the follower's trajectory.  

Problem \eqref{opt: l-lqr} is a Stackelberg game, also known as a \emph{bilevel optimization} problem. See \cite{dempe2020bilevel} and references therein for a detailed discussion of bilevel optimization.

\section{Active inverse learning via difference maximization}
\label{sec: dist max}
Given the Stackelberg trajectory game introduced in Section~\ref{sec: stackelberg}, we now consider the case where the leader does not know the parameter tuple \((Q^\fow, R^\fow, M^\fow)\) in the follower's objective, except that it is one of finitely many candidates. In other words, the leader knows that there exist 
\(Q^1, \ldots, Q^d\in\mathbb{R}^{n^\fow\times n^\fow}\), \(R^1, \ldots, R^d\in\mathbb{R}^{m^\fow\times m^\fow}\), and \(M^1, \ldots, M^d\in\mathbb{R}^{n^\fow\times n^\lea}\) such that
\begin{equation}\label{eqn: hypo}
     (Q^\fow, R^\fow, M^\fow)= (Q^i, R^i, M^i)
\end{equation}
for some \(i\in[1, d]\). This case arises, for example, when a robot has already learned different types of human behavior offline but needs to determine the type of a newly encountered human via online interaction. In the following, we let
\(\theta^\fow\coloneqq (Q^\fow, R^\fow, M^\fow)\) and \(\theta^i\coloneqq (Q^i, R^i, M^i)\)
for all \(i\in[1, d]\). We say that \emph{hypothesis \(i\) is true} if \eqref{eqn: hypo} holds.

Based on a prior probability distribution of all hypotheses that gives the value of \(\mathds{P}(\theta^\fow=\theta^i|x_0^\fow)\) for all \(i\in[1, d]\) and the follower's trajectory \(x^\fow_{1:\tau}\), the leader can infer whether hypothesis \(i\) is more likely to be true than hypothesis \(j\) by computing the following ratio:
\begin{equation}\label{eqn: Bayes rule}
    \frac{\mathds{P}(\theta^\fow=\theta^i|x^\fow_{0:\tau})}{\mathds{P}(\theta^\fow=\theta^j|x^\fow_{0:\tau})}=\frac{\mathds{P}(\theta^\fow=\theta^i|x^\fow_{0})\mathds{P}(x^\fow_{1:\tau}|\theta^\fow=\theta^i)}{\mathds{P}(\theta^\fow=\theta^j|x^\fow_{0})\mathds{P}(x^\fow_{1:\tau}|\theta^\fow=\theta^j)} .
\end{equation}
If the ratio in \eqref{eqn: Bayes rule} is greater than one, then trajectory \(x^\fow_{0:\tau}\) is more likely to occur under hypothesis \(i\) than under hypothesis \(j\), and vice versa.

However, observing the follower's trajectory can be uninformative for the inference above if the trajectories under different hypotheses are similar. For example, suppose that
\begin{equation}\label{eqn: similar response}
    \mathds{P}(x^\fow_{1:\tau}|\theta^\fow=\theta^i)\approx\mathds{P}(x^\fow_{1:\tau}|\theta^\fow=\theta^j)
\end{equation}
for some \(i\neq j\); then, \eqref{eqn: Bayes rule} implies that \( \frac{\mathds{P}(\theta^\fow =\theta^i|x^\fow_{0:\tau})}{\mathds{P}(\theta^\fow=\theta^j|x^\fow_{0:\tau})} \approx \frac{\mathds{P}(\theta^\fow=\theta^i|x^\fow_{0})}{\mathds{P}(\theta^\fow=\theta^j|x^\fow_{0})}\). In other words, observing the follower's ongoing trajectory \(x^\fow_{1:\tau}\) does not help the leader distinguish hypothesis \(i\) from \(j\). In the following, we discuss how the leader can actively avoid the scenarios where \eqref{eqn: similar response} happens. % by making the follower's trajectories under different hypotheses as different as possible.

\subsection{The follower's trajectories under different hypotheses}

To avoid the scenarios where \eqref{eqn: similar response} happens, it suffices to make the follower's trajectory distributions under different hypotheses as different as possible. To this end, we first take a closer look at the follower's trajectory. Proposition~\ref{prop: follower riccati} shows that the follower's state at each time is a Gaussian random variable. Particularly, let
\begin{subequations}\label{eqn: hypothetical response}
    \begin{align}
     & E^i_t= A^\fow-F^i_t P^i_{t+1}A^\fow,\\
    & F^i_t=B^\fow(R^i+(B^\fow)^\top P^i_{t+1}B^\fow)^{-1} (B^\fow)^\top,\\
    &P_{\tau}^i=Q^i,\, P^i_t = Q^i+(A^\fow)^\top P^i_{t+1} E^i_t,\\
    & \Lambda_0^i=0_{n_\fow\times n_\fow}, \,\Lambda_{t+1}^i= E^i_t\Lambda_t^i(E^i_t)^\top +F^i_t+\Omega^\fow,\\
     & q_\tau^i  =-Q^i M^i x_\tau^\lea, \\
     & q_t^i  = (E^i_t)^\top q_{t+1}^i-Q^i M^i x_t^\lea,\\
        & \xi_0^i  =\hat{x}^\fow_0, \, \xi_{t+1}^i = E^i_t\xi_t^i +F^i_t q_{t+1}^i, 
    \end{align}
\end{subequations}
for all \(t\in[0, \tau-1]\). 
Proposition~\ref{prop: follower riccati} implies that, if hypothesis \(i\) is true, \ie, \(\theta^\fow=\theta^i\), then
\begin{equation}\label{eqn: traj Gaussian}
x^\fow_t\sim \mathcal{G}_t^i\coloneqq \mathcal{N}(\xi^i_t, \Lambda^i_t).
\end{equation}

To measure the differences between the trajectory distribution under different hypotheses, we introduce a distance function. To this end, let
\begin{equation}
    \mathbb{D}\coloneqq \{(i, j)~|~ i<j, \, i, j\in[1, d]\}.
\end{equation}
%A popular measure of the difference between two Gaussian distributions is the \emph{KL-divergence}. Given \((i, j)\in\mathbb{D}\) and \(t\in[1, \tau]\), let \(\mathcal{G}_t^i\) and \(\mathcal{G}_t^j\) be defined as in \eqref{eqn: traj Gaussian}. The KL-divergence from \(\mathcal{G}^i_t\) to \(\mathcal{G}^j_t\) is as follows:
% \begin{equation}\label{eqn: KL div}
% \begin{aligned}
%       D_{KL}(\mathcal{G}^i_t || \mathcal{G}^j_t) \coloneqq & \textstyle\frac{1}{2}\norm{\xi_t^i-\xi_t^j}^2_{(\Lambda^j_t)^{-1}}-\frac{(\tau+1)n_\fow}{2}\\
%     &\textstyle +\frac{1}{2}\log \left(\frac{\det \Lambda_t^j}{\det \Lambda_t^i}\right)+\tr((\Lambda_t^j)^{-1}\Lambda_t^i).
% \end{aligned}
% \end{equation}
%Notice that, since \(\Omega^\fow\in\mathbb{S}_{\succ 0}^{n_\fow}\), \eqref{eqn: hypothetical response} implies that \(\Lambda_t^i\in\mathbb{S}_{\succ 0}^{n_\fow}\) for all \(t\in[1, \tau]\) and \(k\in[1, d]\).  However, KL divergence is not symmetric, \ie, \(D_{KL}(\mathcal{G}^i_t || \mathcal{G}^j_t)\neq D_{KL}(\mathcal{G}^j_t || \mathcal{G}^i_t)\). \david{can we just axe all the KL stuff? I see the connection you make later, but maybe it makes sense to just state (16) and then explain it as a symmetric KL distance with some terms dropped. not a huge deal either way} To define a symmetric distance function, 
We propose the following function
\begin{equation}\label{eqn: distance}
    D(\mathcal{G}^i_t, \mathcal{G}^j_t) \coloneqq  \norm{\xi^i_t-\xi^j_t}^2_{(\Lambda^i_t)^{-1}+(\Lambda^j_t)^{-1}}
\end{equation}
for all \(t\in[1, \tau]\) and \((i, j)\in\mathbb{D}\). Notice that, since \(\Omega^\fow\in\mathbb{S}_{\succ 0}^{n_\fow}\), \eqref{eqn: hypothetical response} implies that \(\Lambda_t^i\in\mathbb{S}_{\succ 0}^{n_\fow}\) for all \(t\in[1, \tau]\) and \(k\in[1, d]\). One can verify that the distance function in \eqref{eqn: distance} is proportional to the sum of the KL-divergence from \(\mathcal{G}^i_t\) to  \(\mathcal{G}^j_t\) and the KL-divergence from \(\mathcal{G}^j_t\) to  \(\mathcal{G}^i_t\), up to some additive constants.  Later, we use this function to optimize the leader's inputs.

The intuition behind this distance function is to first evaluate (up to a constant of \(2\), for the convenience of notation) the sum of \(D_{KL}(\mathcal{G}^i_t || \mathcal{G}^j_t)\) and \(D_{KL}(\mathcal{G}^j_t || \mathcal{G}^i_t)\), then remove the terms that are independent of \(\xi^i_t\) and \(\xi^j_t\), which are, as suggested by \eqref{eqn: hypothetical response}, independent of the leader's trajectory. 

\subsection{Maximizing the worst-case pairwise distance}
To avoid the scenarios where \eqref{eqn: similar response} happens, we need to maximize the value of the distance function in \eqref{eqn: distance} for any \((i, j)\in\mathbb{D}\). To this end, we define the following \emph{worst-case distance function}, which evaluates the minimum value of function \eqref{eqn: distance} among all \((i, j)\in\mathbb{D}\): 
\begin{equation}\label{eqn: worst-case dist}
\begin{aligned}
    & \textstyle \underset{(i, j)\in\mathbb{D}}{\min} \left\{ \sum_{t=1}^\tau  \norm{\xi_t^i-\xi_t^i}_{ (\Lambda^i_t)^{-1}+(\Lambda^j_t)^{-1}}^2\right\}\\
     &=\textstyle   \sum\limits_{(i, j)\in\mathbb{D}}  \sum_{t=1}^\tau \norm{\xi_t^i-\xi_t^i}_{(\Lambda^i_t)^{-1}+(\Lambda^j_t)^{-1}}^2\\
     &\textstyle -\underset{(i, j)\in\mathbb{D}}{\max} \left\{ \sum\limits_{(k, l)\in\mathbb{D}\setminus \{(i, j)\}}    \sum_{t=1}^\tau\norm{\xi_t^k-\xi_t^l}_{(\Lambda^k)^{-1}_t+(\Lambda^l)^{-1}_t}^2\right\}.
\end{aligned} 
\end{equation}
The second step in \eqref{eqn: worst-case dist} uses the fact that, given any \(\alpha_1, \ldots, \alpha_n\in\mathbb{R}\), we have 
\begin{equation*}
    \textstyle \min\limits_{i\in[1, n]} \alpha_i = \sum\limits_{i\in[1, n]} \alpha_i-\max\limits_{i\in[1, n]}  \sum\limits_{j\in[1, n], j\neq i}\alpha_j.
\end{equation*}
Notice that the distance in \eqref{eqn: worst-case dist} does not include the terms for \(t=0\) because, due to \eqref{eqn: hypothetical response}, \(\xi_0^i=\xi_0^j\) for all \((i, j)\in\mathbb{D}\).

\begin{figure*}[!hbt]
%\rule[1ex]{\textwidth}{0.1pt}
\begin{equation}\label{opt: dcp}
    \begin{array}{ll}    \underset{s, u_{0:\tau-1}^\lea}{\mbox{minimize}} & s-\sum_{(i, j)\in\mathbb{D}}  \sum_{t=0}^\tau\norm{\xi_t^i-\xi_t^j}_{(\Lambda^i_t)^{-1}+(\Lambda^j_t)^{-1}}^2 +  g(u_{0:\tau-1})\\
    \mbox{subject to}
    & \eta_{t+1}^\lea=A^\lea \eta_t^\lea+B^\lea u_t^\lea,\, \eta_0^\lea=\hat{x}^\lea_0,\\
    &q_t^i  = (E_t^i)^\top q_{t+1}^i-Q^i M^i \eta_t,\, q_\tau^i  =-Q^i M^i \eta_\tau, \, \xi_{t+1}^i = E_t^i\xi_t^i -F_t^i q_{t+1}^i,\,  \xi_0^i  =\hat{x}_0^\fow, \\
    &  \sum_{(k, l)\in\mathbb{D}\setminus\{ (i, j)\}} \sum_{t=1}^\tau  \norm{\xi_t^k-\xi_t^l}_{ (\Lambda^k)^{-1}_t+(\Lambda^l)^{-1}_t}^2\leq s,\, \forall t\in[0, \tau-1], \, i\in[1, d],\, (i, j)\in\mathbb{D}.
    \end{array}
\end{equation}
\rule[1ex]{\textwidth}{0.1pt}
\end{figure*}

Based on the worst-case distance in \eqref{eqn: worst-case dist}, we propose optimizing the leader's input trajectory \(u^\lea_{0:\tau-1}\) by solving optimization \eqref{opt: dcp} (see next page) instead of optimization~\eqref{opt: l-lqr}. In ~\eqref{opt: dcp}, \(E_t^i, F_t^i\) and \(\Lambda_t^i\) are given by \eqref{eqn: hypothetical response}. Also, we approximate the leader's state trajectory \(x^\lea_{0:\tau}\) with its expectation, denoted by \(\eta_{0:\tau}\), which satisfies the averaged dynamics \(\eta_{t+1}=A^\lea \eta_t+B^\lea u^\lea_t\). Notice that this approximation replaces the disturbance term \(w_t^\lea\) in \eqref{opt: l-lqr} with its mean, given by \(w_t^\lea=0_{n_\lea}\) in \eqref{eqn: gaussian disturbance}. 

Additionally, optimization~\eqref{opt: dcp} replaces  \(\mathds{E}[f(x^\lea_{0:\tau}, x^\fow_{0:\tau})]\) in \eqref{opt: l-lqr} with the negative of the worst-case distance in \eqref{eqn: worst-case dist}. In particular, the constraints in \eqref{opt: dcp} imply that
\begin{equation*}\label{eqn: slack ub}
    \textstyle \underset{(i, j)\in\mathbb{D}}{\max} \left\{\sum\limits_{(k, l)\in\mathbb{D}\setminus \{ (i, j)\}}   \sum\limits_{t=1}^\tau \norm{\xi_t^k-\xi_t^l}_{(\Lambda^k)^{-1}_t+(\Lambda^l)^{-1}_t}^2\right\}\leq s.
\end{equation*}
Since the objective function in \eqref{opt: dcp} minimizes the value of \(s\), the above inequality holds as an equality at optimality. Hence, the objective function in \eqref{opt: dcp} is equivalent to the one in \eqref{opt: l-lqr}, except we replace \(\mathds{E}[f(x^\lea_{0:\tau}, x^\fow_{0:\tau})]\) in \eqref{opt: l-lqr} with the negative of the worst-case distance in \eqref{eqn: worst-case dist}. The idea of this replacement is to maximize the worst-case distance in \eqref{eqn: worst-case dist}, ensuring that the distance in \eqref{eqn: distance} is large for any \((i, j)\in\mathbb{D}\). Thus, any pair of hypotheses are easily distinguished.

Optimization ~\eqref{opt: dcp} is a difference-of-convex program: all of its constraints are convex, but its objective function is the difference between two convex functions \cite{horst1999dc}. A popular solution method for difference-of-convex programs is the \emph{convex-concave procedure}, which guarantees global convergence to a stationary point \cite{lanckriet2009convergence} and provides locally optimal solutions  in practice  \cite{lipp2016variations}.

\section{Numerical Experiments and Simulations}

We empirically demonstrate our results using two receding-horizon repeated trajectory games between a boundedly rational follower and a rational leader. 
In the first game, a ground rover controlled by the follower pursues one of \(d=3\) rovers controlled by the leader. At each time step, the agents play a Stackelberg trajectory game and implement only the first step of their respective input trajectories. The leader infers which leading rover the follower is following. In the second game, the leader (a driving assistant) recommends a trajectory for the follower's rover to track. The leader must infer the follower's driving type from \(d=3\) known possibilities (e.g., someone who drives fast, drives slow, or obeys the speed limit) based on how the follower responds to the leader's suggestion. 

We solve the optimization problems for each game in Julia \cite{bezanson2017julia} using the JuMP \cite{lubin2023JUMP} and MOSEK \cite{mosek2024} tools. Then, we convert the optimized trajectories into splines and transmit them to simulated TurtleBots \cite{turtlebot3} in ROS/Gazebo via virtual communication ports \cite{levy2024rossocketsJulia, levy2024rossockets}. The TurtleBots follow the trajectory splines using proportional gain controllers that adjust their linear and angular velocity over time. Since the TurtleBots' speed and maneuverability are limited by their simulated dynamics in Gazebo, we tune the game parameters to ensure that the optimized trajectories are feasible, as described in the following subsections. Note that we run ROS Noetic on Ubuntu Jammy (22.04) using the RoboStack virtual environment \cite{FischerRAM2021}.

\subsection{Leader-follower pursuit game parameters}
We define the parameters for the leader-follower pursuit game as follows. We model the follower's dynamics using an instance of \eqref{sys: follower lti}, where \(A^\fow = \exp\left(\delta\begin{bsmallmatrix}
    0_{2\times 2} & I_2\\
    0_{2\times 2} & 0_{2\times 2}
    \end{bsmallmatrix}\right)\), \(B^\fow=\int_{0}^\delta \exp\left(t\begin{bsmallmatrix}
    0_{2\times 2} & I_2\\
    0_{2\times 2} & 0_{2\times 2}
    \end{bsmallmatrix}\right)\mathrm{d}t \begin{bsmallmatrix}
    0_{2\times 2}\\
    I_2
    \end{bsmallmatrix}\), \(\Omega^\fow=1\times 10^{-5} I_4\), and \(\delta=2\) seconds is the discretization step size. We model the leader's dynamics using the joint dynamics of \(d\in\mathbb{N}\) double-integrators, where \(A^\lea = I_d\otimes A^\fow\), \(B^\lea =I_d\otimes B^\fow\), and \(\Omega^\lea=1\times 10^{-5} I_{4d\times 4d}\). For the follower's trajectory optimization in \eqref{opt: f-lqr}, we let \(Q^i = 20\diag(\begin{bsmallmatrix} 1 & 1& 0 & 0\end{bsmallmatrix} )\), \(R^i =3\times10^4 I_2\), and \(M^i = \begin{bsmallmatrix}
    0_{4\times 4(i-1)} & I_4 & 0_{4\times 4(d-i)}
    \end{bsmallmatrix}\) for all \(i\in[1, d]\). For the leader's trajectory optimization in \eqref{opt: l-lqr}, we let \(\mathbb{U}\coloneqq\{u\in\mathbb{R}^{2d}~|~\norm{u}_\infty\leq 5\times 10^{-3}\}\), and \(g(u^\lea_{0:\tau-1})=\sum_{t=1}^{\tau-2}\norm{u^\lea_{t+1}-u^\lea_t}^2\). Additionally, we constrain the velocity components of each leading rover's trajectory as follows to ensure dynamic feasibility in the simulation:
    \begin{equation*}
        \norm{\begin{bmatrix}
            0_{2\times2} & I_2 \\
            0_{2\times2} & I_2 \\
            0_{2\times2} & I_2
        \end{bmatrix} x^\lea_t}_\infty \leq 0.1 \unit[per-mode = symbol]{\metre\per\second}.
    \end{equation*}
    
    %the velocity magnitudes of each leading rover's trajectory to stay below 0.1 meters per second, ensuring that the optimized trajectories are feasible in simulation: 
    %\(|x^\lea_{t}[4i-1]| \leq 0.1\) and \(|x^\lea_{t}[4i]| \leq 0.1\) for all \(i \in [1,d]\).

\subsection{Driving assistant game parameters} \label{subsec: driving parameters}
Similarly, in the driving assistant game, we model the follower's dynamics using an instance of \eqref{sys: follower lti}, where \(A^\fow = \exp\left(\delta\begin{bsmallmatrix}
    0_{2\times 2} & I_2\\
    0_{2\times 2} & 0_{2\times 2}
    \end{bsmallmatrix}\right)\), 
\(B^\fow=\int_{0}^\delta \exp\left(t\begin{bsmallmatrix}
    0_{2\times 2} & I_2\\
    0_{2\times 2} & 0_{2\times 2}
    \end{bsmallmatrix}\right)\mathrm{d}t \begin{bsmallmatrix}
    0_{2\times 2}\\
    I_2
    \end{bsmallmatrix}\), \(\Omega^\fow=1\times10^{-4}I_4\), and \(\delta=2\) seconds. We model the leader's dynamics using an instance of \eqref{sys: leader lti}, where \(A^\lea = A^\fow\), \(B^\lea = B^\fow\), and \(\Omega^\lea=1\times10^{-7}I_4\). For the follower's trajectory optimization in \eqref{opt: f-lqr}, we let \(Q^i = \diag(\begin{bsmallmatrix} 1000 & 100& 100 & 100\end{bsmallmatrix} )\) and \(R^i =\diag(\begin{bsmallmatrix} 1\times10^4 & 1000\end{bsmallmatrix})\).  Then, we model three follower driver types, each with a unique \(M^i\). Given the leader's suggested trajectory, type 1 scales down the velocity (\(M^1 = \diag(\begin{bsmallmatrix}
    0.95 & 0.85 & 0.95 & 0.85
    \end{bsmallmatrix})\), type 2 follows the suggested velocity (\(M^2 = I_4\)), and type 3 scales up the velocity (\(M^3 = \diag(\begin{bsmallmatrix}
    1.05 & 1.15 & 1.05 & 1.15
    \end{bsmallmatrix})\). For the leader's trajectory optimization in \eqref{opt: l-lqr}, we let \(\mathbb{U}\coloneqq\{u\in\mathbb{R}^{2d}~|~\norm{u}_\infty\leq 5\times 10^{-2}\}\), and \(g(u^\lea_{0:\tau-1})=\sum_{t=1}^{\tau-2}\norm{u^\lea_{t+1}-u^\lea_t}^2\).

    Additionally, we constrain the leader and follower trajectories within the boundaries of a sideways ``L'' shaped road.  The road width \(r_w = 1.3\) meters, and the road length of the short segment \(r_l = 3\) meters. The follower should neither drive backwards nor drive off the road. Therefore, during the first road segment, for all \(t\in[1,\frac{\tau}{2}]\), we apply the following constrains: \(\begin{bsmallmatrix}
            0 & 1 & 0 & 0
        \end{bsmallmatrix} x^\lea_t \geq 0\) m, 
        \(\begin{bsmallmatrix}
            0 & 0 & 0 & 1
        \end{bsmallmatrix} x^\lea_t \geq 0 \) m/s, and
        \(\norm{\begin{bsmallmatrix}
            1 & 0 & 0 & 0
        \end{bsmallmatrix} x^\fow_t}_\infty \leq 0.5r_w \). Similarly, during the second road segment, for all \(t\in[\frac{\tau}{2},\tau]\), \(\begin{bsmallmatrix}
            1 & 0 & 0 & 0
        \end{bsmallmatrix} x^\lea_t \geq -0.5r_w \),
        \(\begin{bsmallmatrix}
            0 & 0 & 1 & 0 
        \end{bsmallmatrix}x^\lea_t \geq 0 \) m/s, and
        \(r_l \leq \begin{bsmallmatrix}
            0 & 1 & 0 & 0
        \end{bsmallmatrix}x^\fow_t \leq r_l + r_w \).

    %Therefore, during the first road segment, we constrain the leader's \(y\) position and \(y\) velocity to be non-negative, and we constrain the follower's \(x\) position to stay between \(-0.5r_w\) and \(0.5r_w\). Similarly during the second road segment, we constrain the leader's \(x\) position to be greater than or equal to \(-0.5r_w\) and the \(x\) velocity to be non-negative. We also constrain the follower's \(y\) position to stay between \(r_l\) and \(r_l + r_w\).

\subsection{Experimental procedure}
    We demonstrate the proposed learning methods using the simulation of two receding-horizon repeated trajectory games. At time \(t\delta\), for some \(t\in\mathbb{N}\), the leader observes its current state \(\overline{x}^\lea_t\) and the follower's current state \(\overline{x}^\fow_t\). Then, the leader solves optimization~\eqref{opt: dcp} with \(\hat{x}_0^\lea=\overline{x}^\lea_t\) and \(\hat{x}^\fow_0=\overline{x}^\fow_t\) to obtain the optimal input sequence \(u^\lea_{0:\tau-1}\). Next, the leader simulates a state trajectory \(x^\lea_{0:\tau-1}\) according to \eqref{sys: leader lti},  shares the trajectory with the follower, and applies \(u^\lea_0\) to its system. Meanwhile, at time \(t\delta\), the follower observes its current state \(\hat{x}_t^\fow\) and receives the leader's simulated trajectory \(x_{0:\tau}^\lea\). The follower solves optimization~\eqref{opt: f-lqr} with \(\hat{x}^\fow_0=\overline{x}^\fow_t\) and obtains the optimal mean and covariance sequences \((\mu_{0:\tau-1}, \Sigma_{0:\tau-1})\). Finally, the follower samples \(u_0^\fow\sim\mathcal{N}(\mu_0, \Sigma_0)\) and applies \(u_0^\fow\) to its system.

\subsection{Results}
     % in these receding-horizon repeated trajectory games
    We simulate the players' trajectories by solving the leader's trajectory optimization~\eqref{opt: dcp} using the convex-concave procedure \cite{lipp2016variations}. Fig.~\ref{subfig: pursuit traj} shows the results of the leader-follower pursuit game. To maximize the differences between the follower's trajectories under different hypotheses, optimization~\eqref{opt: dcp} ensures that different leading rovers move in different directions. Fig.~\ref{subfig: driver traj} shows the results of the driving assistant game. Optimization~\eqref{opt: dcp} ensures that the leader suggests a trajectory to the follower that clearly distinguishes the follower's response under each driver type. 

    We further showcase the advantage of the proposed method in terms of distinguishing different hypotheses. Let \(p^\star= \begin{bmatrix}
    1 & 0 & \cdots & 0
    \end{bmatrix}^\top\in\mathbb{R}^{d}\) denote the ground truth probability of all hypotheses. Without loss of generality, we assume \((Q^\fow, R^\fow, M^\fow)=(Q^1, R^1, M^1)\). In addition, let \(p^t\) denote the \(d\)-dimensional vector such that
    \begin{equation}
        p_i^t\coloneqq \mathds{P}((Q^\fow, R^\fow, M^\fow)= (Q^i, R^i, M^i)~|~\overline{x}^\fow_{0:t})
    \end{equation}
    for all \(i\in[1, d]\), where we let \(\mathds{P}((Q^\fow, R^\fow, M^\fow)= (Q^i, R^i, M^i)~|~\overline{x}^\fow_{0})\coloneqq \frac{1}{d}\) for all \(i\in[1, d]\). That is, we choose the uniform distribution as the leader's prior distribution over all hypotheses given the follower's initial state. Notice that one can compute \(p^t\) recursively using Bayes rule. 
    
    Fig.~\ref{fig: belief distribution} shows the time history of the \(\ell_1\)-norm distance between \(p^t\) and \(p^\star\) and compares the results to where the leader uses randomly sampled trajectories. We generate the random trajectories in two steps. First, we sample a trajectory from a uniform distribution over a set of allowable states. Second, we solve an instance of optimization ~\eqref{opt: f-lqr}, where the leader follows the random trajectory as close as possible without violating it's dynamics.  %In the pursuit game, since the leader's dynamics are the joint dynamics of three rovers, we solve optimization ~\eqref{opt: f-lqr} once for each follower type and stack the trajectories into one matrix to form the leader's complete trajectory. %we substitute \(x_t^\fow = x_t^\lea\) and \(M^\fow x_t^\lea = x_{t,rand}^\lea\).
    
    In the pursuit game, we constrain the random positions within \(\pm 0.3 \unit{m}\) and the velocities within \(\pm 0.1 \unit[per-mode = symbol]{\metre\per\second} \). In the driving assistant game, we constrain the random positions within the road boundaries (defined in section \ref{subsec: driving parameters}). In the first road segment, we constrain the horizontal and vertical velocities, respectively, within \(\pm 1.3\) and \((0,0.5)\) \(\unit[per-mode = symbol]{\metre\per\second} \). Similarly, in the second road segment, we constrain the velocities within \((0, 1.3)\) and \(\pm 0.5\) \(\unit[per-mode = symbol]{\metre\per\second} \).

    The results in Fig.~\ref{fig: belief distribution} show that \(p^t\) converges faster to \(p^\star\) when the leader uses our proposed method, versus randomly generated trajectories. In the pursuit game, the proposed method consistently achieves orders of magnitude better convergence than the random trajectories. However, in the driving game, the proposed method only shows significant gains after the first six seconds. The driving game is a lower dimensional problem than the pursuit game, and the road boundaries add additional state constraints. As a result, it is more difficult to distinguish between the hypotheses. 

    % In both games, the proposed method eventually achieves better convergence than the random method. 
    % In the pursuit game, the proposed method significantly outperforms the random method. 
    % In the driving assistant game, convergence is slower, and it takes a few seconds for the proposed method to show improvements over random inputs. The driving game is a lower dimensional problem with more state constraints than the pursuit game, making it more difficult to distinguish the hypotheses. 
    
     We simulate the experiments in Gazebo (see Fig.~\ref{fig: gazebo}). We provide a video of our Gazebo simulations at \url{https://youtu.be/csQpXJh1SmM} and our code at \url{https://github.com/u-t-autonomous/active_inverse_stackelberg} and \url{https://github.com/willward20/active_inverse_stackelberg_ros}.
     
\captionsetup[figure]{belowskip=3pt}
\begin{figure}[!t]
\centering
\begin{subfigure}{0.49\columnwidth}
\centering
\begin{tikzpicture}
            \begin{axis}[
                xlabel={\footnotesize horizontal position },
                ylabel={\footnotesize vertical position },
                xmin=-3, xmax=3.5,
                ymin=-2.5, ymax=3.5,
                width=\columnwidth,
                height =\columnwidth, 
                xtick={-2, 0, 2},
                xtick style={draw=none},
                ytick={-2, 0, 2},
                ytick style={draw=none},
                legend pos=south west,
                legend style={font=\footnotesize},
                yticklabel style = {font=\footnotesize},
                xticklabel style = {font=\footnotesize},
                ymajorgrids = true,
                xmajorgrids = true
            ]

                \pgfplotsextra{\DTLforeach*{LF_type1}{\cx=x, \cy=y, \rx=l1, \ry=l2}{
                    \draw[blue!15, fill=blue!15] (\cx,\cy) ellipse [x radius=\rx, y radius=\ry];
                }}
                \pgfplotsextra{\DTLforeach*{LF_type2}{\cx=x, \cy=y, \rx=l1, \ry=l2}{
                    \draw[green!15, fill=green!15] (\cx,\cy) ellipse [x radius=\rx, y radius=\ry];
                }}
                \pgfplotsextra{\DTLforeach*{LF_type3}{\cx=x, \cy=y, \rx=l1, \ry=l2}{
                    \draw[red!15, fill=red!15] (\cx,\cy) ellipse [x radius=\rx, y radius=\ry];
                }}
                
                % Plot the leader rover trajectories
                \addplot [mark=*, darkgray, mark size=0.5pt] table [x=x1, y=y1, col sep=space]{simulations/pursuit_traj.dat};
                \addplot [mark=*, darkgray,mark size=0.5pt] table [x=x2, y=y2, col sep=space]{simulations/pursuit_traj.dat};
                \addplot [mark=*, darkgray, mark size=0.5pt] table [x=x3, y=y3, col sep=space]{simulations/pursuit_traj.dat};

                % Plot the mean follower trajectories
                \addplot [mark=*, blue, mark size=0.5pt] table [x=x, y=y, col sep=space]{simulations/traj_plot/LF_type1.dat};

                \addplot [mark=*, green, mark size=0.5pt] table [x=x, y=y, col sep=space]{simulations/traj_plot/LF_type2.dat};

                \addplot [mark=*, red, mark size=0.5pt] table [x=x, y=y, col sep=space]{simulations/traj_plot/LF_type3.dat};

                % % Plot follower 1's bundles
                % \addplot [yyred, forget plot, mark=none, name path=r1, opacity=0.1] table [x=r1x, y=r1y, col sep=space]{simulations/pursuit_bundles.dat};
                % \addplot [yyred, forget plot, mark=none, name path=r2, opacity=0.1] table [x=r2x, y=r2y, col sep=space]{simulations/pursuit_bundles.dat};
                % \addplot[yyred, forget plot, opacity=0.15] fill between[of=r1 and r2];
                
                % % Plot follower 2's bundles
                % \addplot [yygreen, forget plot, mark=none, name path=g1, opacity=0.1] table [x=g1x, y=g1y, col sep=space]{simulations/pursuit_bundles.dat};
                % \addplot [yygreen, forget plot, mark=none, name path=g2, opacity=0.1] table [x=g2x, y=g2y, col sep=space]{simulations/pursuit_bundles.dat};
                % \addplot[yygreen, forget plot, opacity=0.15] fill between[of=g1 and g2];
    
                % % Plot follower 3's bundles
                % \addplot [yyblue, forget plot, mark=none, name path=b1, opacity=0.1] table [x=b1x, y=b1y, col sep=space]{simulations/pursuit_bundles.dat};
                % \addplot [yyblue, forget plot, mark=none, name path=b2, opacity=0.1] table [x=b2x, y=b2y, col sep=space]{simulations/pursuit_bundles.dat};
                % \addplot[yyblue, forget plot, opacity=0.15] fill between[of=b1 and b2];
    
            \end{axis}  
        \end{tikzpicture}
\caption{Pursuit game.}
\label{subfig: pursuit traj}
\end{subfigure}
\begin{subfigure}{\columnwidth}
\centering
\begin{tikzpicture}
            \begin{axis}[
                xlabel={\footnotesize horizontal position },
                ylabel={\footnotesize vertical position },
                xmin=-1, xmax=11.5,
                ymin=-0.25, ymax=4.75,
                width=\columnwidth,
                height =0.5*\columnwidth, 
                xtick={0, 2, 4, 6, 8, 10},
                xtick style={draw=none},
                ytick={0, 2, 4},
                ytick style={draw=none},
                legend pos=north east,
                legend style={font=\footnotesize},
                yticklabel style = {font=\footnotesize},
                xticklabel style = {font=\footnotesize},
                ymajorgrids = true,
                xmajorgrids = true
            ]

                \pgfplotsextra{\DTLforeach*{LRoad_type1}{\cx=x, \cy=y, \rx=l1, \ry=l2}{
                    \draw[blue!15, fill=blue!15] (\cx,\cy) ellipse [x radius=\rx, y radius=\ry];
                }}
                 \pgfplotsextra{\DTLforeach*{LRoad_type2}{\cx=x, \cy=y, \rx=l1, \ry=l2}{
                    \draw[green!15, fill=green!15] (\cx,\cy) ellipse [x radius=\rx, y radius=\ry];
                 }}
                 \pgfplotsextra{\DTLforeach*{LRoad_type3}{\cx=x, \cy=y, \rx=l1, \ry=l2}{
                    \draw[red!15, fill=red!15] (\cx,\cy) ellipse [x radius=\rx, y radius=\ry];
                }}

                % Version 1: only plot the leader's trajectory.
                % NOTE: this is also the reference trajectory for follower 2
                \addplot [mark=*, black, mark size=0.75pt] table [x=x, y=y, col sep=space]{simulations/traj_plot/LRoad_leader.dat};

                % Version 2: plot each follower type's reference trajectory
                %\addplot [mark=*, blue, mark size=0.5pt] table [x=rl2x, y=rl2y, col sep=space]{simulations/driver_ref_traj.dat};
                %\addplot [mark=*, blue, mark size=0.5pt] table [x=rl1x, y=rl1y, col sep=space]{simulations/driver_ref_traj.dat};
                %\addplot [mark=o, red, mark size=0.5pt] table [x=rl3x, y=rl3y, col sep=space]{simulations/driver_ref_traj.dat};

                % Plot the mean follower trajectories
                \addplot [mark=*, blue, mark size=0.5pt] table [x=x, y=y, col sep=space]{simulations/traj_plot/LRoad_type1.dat};
                \addplot [mark=*, green, mark size=0.5pt] table [x=x, y=y, col sep=space]{simulations/traj_plot/LRoad_type2.dat};
                \addplot [mark=*, red, mark size=0.5pt] table [x=x, y=y, col sep=space]{simulations/traj_plot/LRoad_type3.dat};

                % Horizontal Road Margins
                \addplot+ [darkgray, dashed, mark=none, const plot,
                           empty line=jump]
                        coordinates {(-0.65, 4.3)
                                     (12, 4.3)
            
                                     (0.65, 3)
                                     (12, 3)
                        };
                % Vertical Road Margins
                \addplot+ [darkgray, dashed, mark=none, const plot,
                           empty line=jump]
                        coordinates {
                                (-0.65, -0.5)
                                (-0.65, 4.3)
                        
                                (0.65, -0.5)
                                (0.65, 3)
                        };
            \end{axis}  
        \end{tikzpicture}
\caption{Driving assistant game.}
\label{subfig: driver traj}
\end{subfigure}
\caption{The trajectories of the leader (black) and followers under different hypotheses. Each color indicates one of the hypotheses. The semi-axis of the shaded ellipses represents the directions of the eigenvectors of the corresponding covariance matrices. }
\label{fig: traj distribution}
\end{figure}

\begin{figure}[!t]
\captionsetup[subfigure]{aboveskip=-5pt,belowskip=0pt}
\begin{subfigure}{0.49\columnwidth}
\centering
\begin{tikzpicture}
    \begin{semilogyaxis}[
            xlabel={{\footnotesize $t$}},
            ylabel={{\footnotesize $\norm{p^t-p^\star}_1$}},
            xmin=1, xmax=9,
            ymin=1e-8, ymax=1e1,
            width=\columnwidth,
            height =\columnwidth, 
            xtick={0, 2, 4, 6, 8},
            xtick pos=left,
            ytick={1e-6, 1e-4, 1e-2, 1e0},
            ytick style={draw=none},
            yticklabel style = {font=\footnotesize},
            xticklabel style = {font=\footnotesize},
            ymajorgrids = true,
            xmajorgrids = true,
    ]

\addplot [mark=none, color=blue,line width=1pt] table[x=t,y=err] {simulations/errplot/LF_rand.dat};
%\addlegendentry{\footnotesize Random}

\addplot [mark=none, color=red,line width=1pt] table[x=t,y=err] {simulations/errplot/LF_opt.dat};
%\addlegendentry{\footnotesize Proposed}

\addplot [name path=upper, draw=none] table[x=t,y=errmax] {simulations/errplot/LF_rand.dat};
\addplot [name path=lower, draw=none] table[x=t,y=errmin] {simulations/errplot/LF_rand.dat};
\addplot [fill=blue!10] fill between[of=upper and lower];

\addplot [name path=upper, draw=none] table[x=t,y=errmax] {simulations/errplot/LF_opt.dat};
\addplot [name path=lower, draw=none] table[x=t,y=errmin] {simulations/errplot/LF_opt.dat};
\addplot [fill=red!10] fill between[of=upper and lower];

 \end{semilogyaxis}
    \end{tikzpicture}    
\caption{Pursuit game.}
\end{subfigure}
\begin{subfigure}{0.49\columnwidth}
\centering
\begin{tikzpicture}
    \begin{semilogyaxis}[
             xlabel={{\footnotesize $t$}},
            ylabel={{\footnotesize $\norm{p^t-p^\star}_1$}},
            xmin=1, xmax=9,
            ymin=1e-8, ymax=1e1,
            width=\columnwidth,
            height =\columnwidth,
            xtick={2, 4, 6, 8},
            xtick pos=left,
            ytick={1e-6, 1e-4, 1e-2, 1e0},
            ytick style={draw=none},
            yticklabel style = {font=\footnotesize},
            xticklabel style = {font=\footnotesize},
            ymajorgrids = true,
            xmajorgrids = true,
    ]

\addplot [mark=none, color=blue,line width=1pt] table[x=t,y=err] {simulations/errplot/LRoad_rand.dat};
%\addlegendentry{\footnotesize Random}

\addplot [mark=none, color=red,line width=1pt] table[x=t,y=err] {simulations/errplot/LRoad_opt.dat};
%\addlegendentry{\footnotesize Proposed}

\addplot [name path=upper, draw=none] table[x=t,y=errmax] {simulations/errplot/LRoad_rand.dat};
\addplot [name path=lower, draw=none] table[x=t,y=errmin] {simulations/errplot/LRoad_rand.dat};
\addplot [fill=blue!10] fill between[of=upper and lower];

\addplot [name path=upper, draw=none] table[x=t,y=errmax] {simulations/errplot/LRoad_opt.dat};
\addplot [name path=lower, draw=none] table[x=t,y=errmin] {simulations/errplot/LRoad_opt.dat};
\addplot [fill=red!10] fill between[of=upper and lower];

 \end{semilogyaxis}
    \end{tikzpicture}    
\caption{Driving assistant game.}
\end{subfigure}
\caption{Comparing the convergence of \(p^t\) when using the proposed method (red) versus sampling a random trajectory (blue). The solid lines show the median over 100 simulations, while the boundaries of the colored areas mark the corresponding first and third quartiles.}
\label{fig: belief distribution}
\end{figure}

\begin{figure}[!t]
\centering
\begin{subfigure}{0.325\columnwidth}
\centering
\includegraphics[width=\columnwidth]{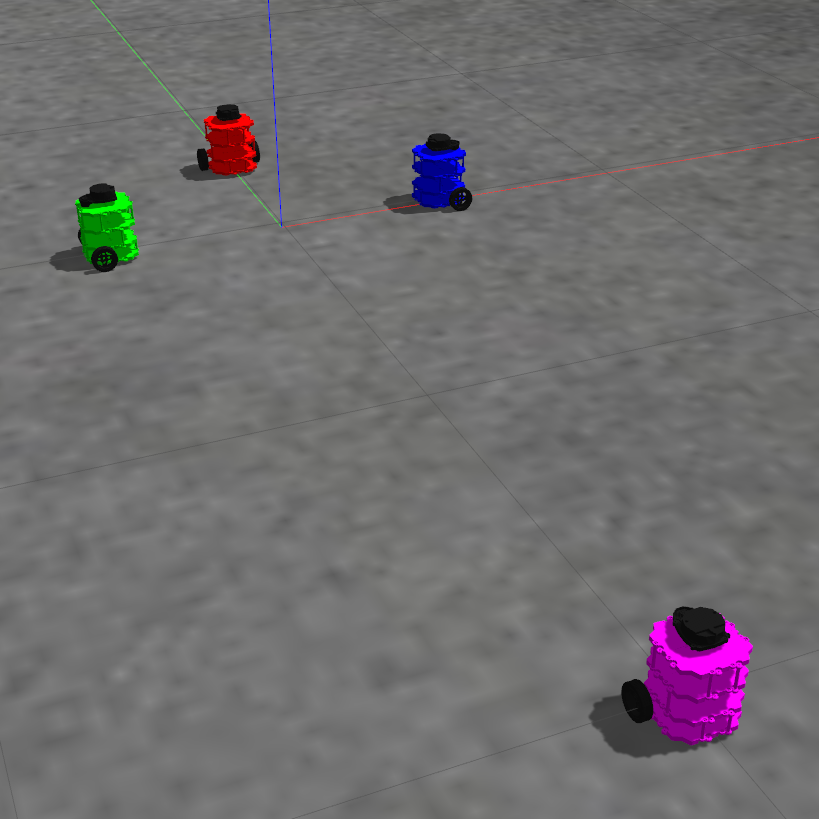}
\caption{Pursuit game.}
\end{subfigure}
\begin{subfigure}{0.65\columnwidth}
\centering
\includegraphics[width=\columnwidth]{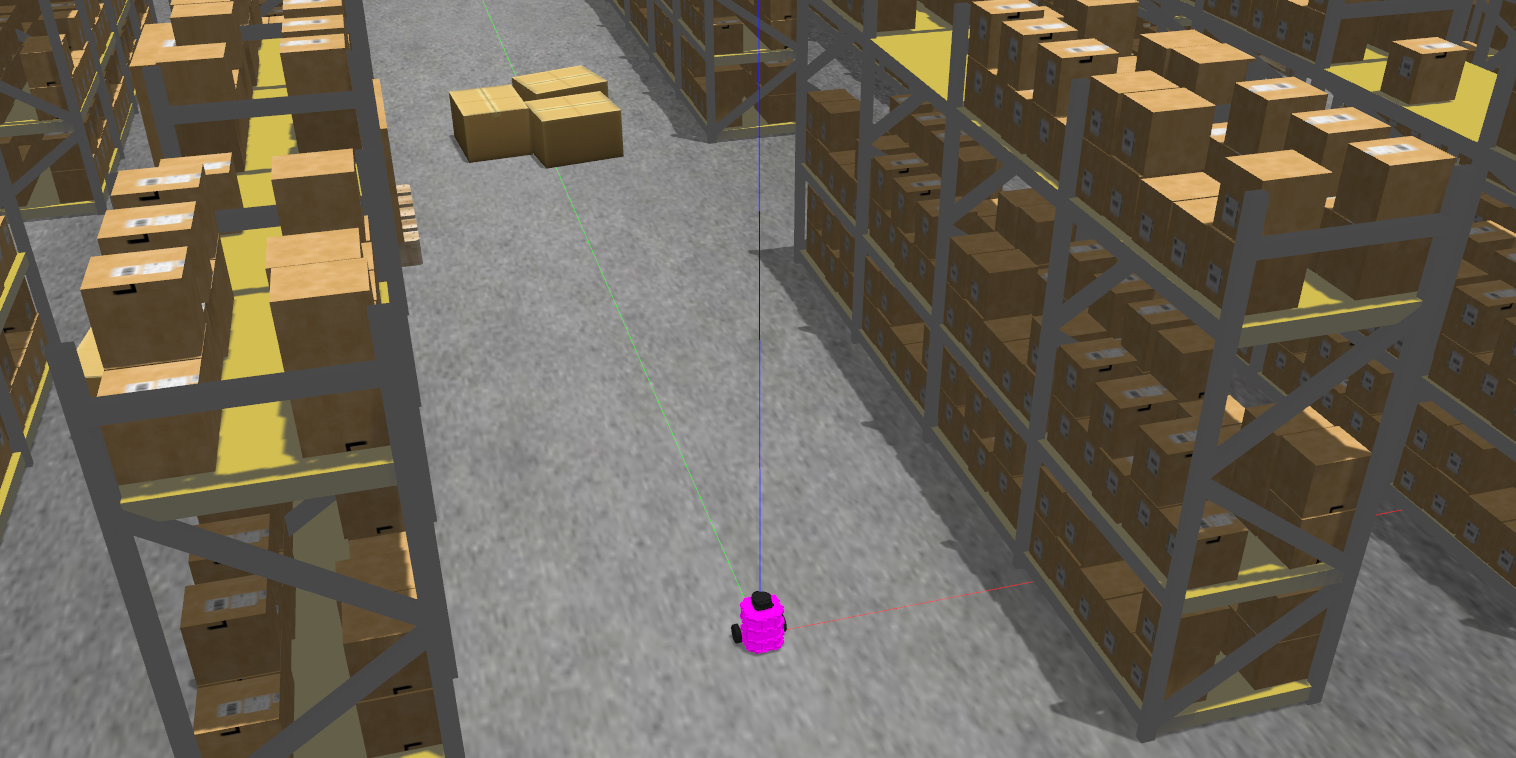}
\caption{Driving assistant game.}
\end{subfigure}
\caption{Simulating the results on TurtleBots in a virtual warehouse environment in Gazebo.}
\label{fig: gazebo}
\end{figure}

\section{Conclusion}

We formulated an inverse learning problem in a Stackelberg trajectory game, where the leader infers the type of the follower's cost function by observing its trajectories. We proposed an active inverse learning method to accelerate the leader inference by making the follower's trajectories under different hypotheses as different as possible. This method provides faster convergence of the probability of different hypotheses when compared against random inputs.

However, the current work still has limitations. For example, it considers neither nonlinear dynamics nor collision avoidance in the players' trajectory optimization. In addition, it ignores the possibility of deceptive actions of the follower. In future work, we plan to address these limitations and develop information-gathering strategies in general Bayesian games.


%%%%%%%%%%%%%%%%%%%%%%%%%%%%%%%%%%%%%%%%%%%%%%%%%%%%%%%%%%%%%%%%%%%%%%%%%%%%%%%%

\section*{APPENDIX}
\subsection*{Proof of Proposition~\ref{prop: follower riccati} }
Given \(y\in\mathbb{R}^{n_\fow}\) and \(t\in[0, \tau]\), let
\begin{equation}\label{eqn: Lyapunov 1}
\begin{aligned}
&\psi_t(y, \mu_{t:\tau}, \Sigma_{t:\tau})\\
&\coloneqq \textstyle  \sum_{j=t}^\tau \mathds{E}\left[\frac{1}{2}\norm{x_j^\fow-M^\fow x_j^\lea}_{Q^\fow}^2|x_t^\fow=y\right]\\
&\textstyle +\frac{1}{2}\sum_{j=t}^{\tau-1} \left(\mathds{E}\left[\norm{u_j^\fow}_{R^\fow}^2|x_t^\fow=y\right]\textstyle - \log\det \Sigma_j\right).
\end{aligned} 
\end{equation}
where \(x_{j+1}^\fow=A^\fow x_j^\fow+B^\fow u_j^\fow\) and \(u_j^\fow\sim\mathcal{N}(\mu_j, \Sigma_j)\). 
Furthermore, let \( V_t(y)\coloneqq \min_{\mu_{t:\tau-1}, \Sigma_{t:\tau-1}} \psi_t(y, \mu_{t:\tau}, \Sigma_{t:\tau})\) for all \(t\in[0, \tau]\). Then one can verify that the optimal value of optimization~\eqref{opt: f-lqr} is \(V(\hat{x}_0^\fow)\). In addition, we can show that \(V_\tau(y)=\frac{1}{2}y^\top Q^\fow y - \langle  Q^\fow M^\fow x_\tau^\lea, y\rangle+\frac{1}{2}\norm{M^\fow x_{\tau}^\lea}^2_{Q^\fow}\), \ie, \(V_\tau(y)\) is a quadratic function of \(y\). Suppose that \(V_{t+1}(y)\) is a quadratic function of \(y\), \ie, there exists \(P_{t+1}^\fow\in\mathbb{R}^{n_\fow\times n_\fow}\), \(q_{t+1}^\fow\in\mathbb{R}^\fow\), and \(\nu_{t+1}^\fow\in\mathbb{R}\), such that \(V_{t+1}(y)=\frac{1}{2}y^\top P_{t+1}^\fow y+\langle q_{t+1}^\fow, y\rangle+\nu^\fow_{t+1}\). Then the principle of dynamic programming together imply that
\begin{equation}\label{eqn: Vt}
\begin{aligned}
    &V_t(y) =  \textstyle \frac{1}{2}\norm{y-M^\fow x_t^\fow}^2_{Q^\fow} -\frac{1}{2}\log\det\Sigma_t\\
    &+\min_{\mu_t, \Sigma_t} \textstyle \mathds{E}\left[\frac{1}{2}\norm{u_t^\fow}_{R^\fow}^2+V_{t+1}(A^\fow x_t^\fow+B^\fow u_t^\fow+w_t^\fow)|x_t^\fow=y\right],
\end{aligned}  
\end{equation}
where \(u_t^\fow|x_t^\fow\sim\mathcal{N}(\mu_t, \Sigma_t)\). Observe that
\begin{equation}\label{eqn: R term}
\begin{aligned}
     &\mathds{E}[\norm{u_t^\fow}_{R^\fow}^2|x_t^\fow=y] \\
    &=\mu_t^\top R^\fow \mu_t+\mathds{E}[\tr ((u_t^\fow-\mu_t)(u_t^\fow-\mu_t)^\top R^\fow) |x_t^\fow=y]\\
    &=\mu_t^\top R^\fow \mu_t +\tr(\Sigma_tR^\fow).
\end{aligned} 
\end{equation}
In addition, we can show that
\begin{equation}\label{eqn: Q term}
    \begin{aligned}
        &\mathds{E}[V_{t+1}(A^\fow x_t^\fow+B^\fow u_t^\fow+w_t^\fow)|x_t^\fow=y]\\
        &= \textstyle\frac{1}{2} (A^\fow y+B^\fow\mu_t)^\top P_{t+1}^\fow(A^\fow y+B^\fow \mu_t)\\
        &+\mathds{E}[(A^\fow x_t^\fow+B^\fow \mu_t)^\top P_{t+1}^\fow B^\fow (u_t^\fow-\mu_t)|x_t^\fow=y]\\
        &\textstyle+\frac{1}{2}\mathds{E}[(B^\fow (u_t^\fow-\mu_t))^\top P_{t+1}^\fow B^\fow (u_t^\fow-\mu_t)|x_t^\fow=y]\\
        &+\langle q_{t+1}^\fow, A^\fow y+B^\fow \mu_t\rangle +\frac{1}{2}\mathds{E}[\tr(w_t^\fow(w_t^\fow)^\top P_{t+1})]+\nu_{t+1}^\fow\\
        &=\textstyle \frac{1}{2} (A^\fow y+B^\fow \mu_t)^\top P_{t+1}^\fow (A^\fow y+B^\fow \mu_t)+\frac{1}{2}\tr(\Omega^\fow P_{t+1})\\
        &\textstyle+\frac{1}{2}\tr(\Sigma_t(B^\fow)^\top P_{t+1}^\fow  B^\fow )+\langle q_{t+1}^\fow, A^\fow y+B^\fow \mu_t\rangle+\nu_{t+1}^\fow.
    \end{aligned}
\end{equation}
Substituting \eqref{eqn: Q term} and \eqref{eqn: R term} into \eqref{eqn: Vt} gives
\begin{equation}\label{eqn: Vt explicit}
    \begin{aligned}
        &V_t(y)=\textstyle+ \frac{1}{2}\tr(\Sigma_t R^\fow + \Omega^\fow P_{t+1}+(B^\fow)^\top P_{t+1}^\fow B^\fow ))+\nu_{t+1}^\fow\\
        &+\textstyle \frac{1}{2}y^\top (Q^\fow +(A^\fow)^\top P_{t+1}^\fow A^\fow )y+ \langle A^\top q_{t+1}^\fow-Q^\fow M^\fow x_t^\lea, y\rangle\\
        &\textstyle+\frac{1}{2}\mu_t^\top (R^\fow +(B^\fow)^\top P_{t+1}^\fow B^\fow)\mu_t+\frac{1}{2}(x_t^\lea)^\top (M^\fow)^\top Q^\fow M^\fow x_t^\lea\\
        &\textstyle + \langle (B^\fow)^\top q_{t+1}^\fow+(B^\fow)^\top P_{t+1}^\fow A^\fow y, \mu_t\rangle-\frac{1}{2}\log\det \Sigma_t.
    \end{aligned}
\end{equation}
By setting the derivative of \(V_t(y)\) with respect to \(\mu_t\) and \(\Sigma_t\) to zero, we obtain \eqref{eqn: mean & var for follower input}. By substituting \eqref{eqn: mean & var for follower input} into \eqref{eqn: Vt explicit}, we can show that \(V_t(y)=\frac{1}{2}y^\top P_t^\fow y+\langle q_t^\fow, y\rangle+\nu_t^\fow\), where \(Q_t^\fow\) and \(q_t^\fow\) satisfy \eqref{eqn: follower DP}.

Next, let \(K_t=-\Sigma_t (B^\fow)^\top P_{t+1}^\fow A^\fow\) and \(b_t = -\Sigma_t (B^\fow)^\top q^\fow_{t+1}\). Then \eqref{eqn: mean & var for follower input} implies that \(u_t^\fow|x_t^\fow\sim\mathcal{N}(K_tx_t^\fow+b_t, \Sigma_t)\). Since \(x_0^\fow=\hat{x}^\fow_0\), by using the results in \cite[p. 91]{bishop2006pattern} we can show the following:
\begin{equation*}
    \begin{bmatrix}
    x_1^\fow\\ u_1^\fow 
    \end{bmatrix} \sim\mathcal{N}\left(
    \begin{bmatrix}
    \xi_0\\
    K_0\xi_0+b_0
    \end{bmatrix}, \begin{bmatrix}
    \Lambda_0 & K_0\Lambda_0\\
  \Lambda_0 K_0^\top & \Sigma_0 +K_0^\top \Lambda_0K_0
    \end{bmatrix}
    \right).
\end{equation*}
Therefore \(x_{t+1}^\fow=A^\fow x_t^\fow+B^\fow u_t^\fow+w_t^\fow \sim\mathcal{N}(\xi_1, \Lambda_1)\), where \(\xi_1\) and \(\Sigma_1\) satisfy \eqref{eqn: follower state traj}. By repeating similar steps for \(t\in [2, \tau]\) we can show that \eqref{eqn: follower state traj} holds for all \(t\in[0, \tau-1]\), which completes the proof.
 \label{sec:appendix}

\section*{ACKNOWLEDGMENT} This work was supported by a National Science Foundation CAREER award under Grant No. 2336840, as well as the following NSF grants: ECCS-2145134, CNS-2218759, and CCF-2211542. Additionally, this work was supported by ONR subcontracts with Phillips (N00014-24-S-B001) and the University of Pennsylvania (N00014-20-1-2115).

\bibliographystyle{IEEEtran}
\bibliography{IEEEabrv,reference.bib}

%\printbibliography

\end{document}